\documentclass[superscriptaddress,amsmath,amssymb,aps,pra,twocolumn,]{revtex4-2}

\usepackage{xcolor}
\usepackage{dcolumn} 
\usepackage{bm,graphics}
\usepackage{dsfont}
\usepackage{mathtools}
\usepackage{subcaption}

\usepackage{amsmath}
\usepackage{amssymb}
\usepackage{graphicx}
\usepackage{amsthm}
\usepackage{bbold}
\usepackage{xcolor}
\usepackage[linesnumbered,ruled,vlined]{algorithm2e}
\usepackage{algorithmic}

\usepackage{times}
\usepackage[normalem]{ulem}
\usepackage{multirow}
 \usepackage{ragged2e}

\usepackage[colorlinks=true,linkcolor=teal, citecolor=teal
]{hyperref}

\usepackage{booktabs}

\usepackage[most]{tcolorbox}
\tcbuselibrary{breakable}
\newtcolorbox[auto counter, number within=section]{mybox}[2][]{
  float, floatplacement=tp,
  colback=blue!5!white, colframe=blue!75!black,
  fonttitle=\bfseries, center title,
  title=Box~\thetcbcounter: #2, #1
}

\newcounter{mybox}

\usepackage[ruled]{algorithm2e}

\newcommand{\ket}[1]{|#1\rangle}
\newcommand{\bra}[1]{\langle#1|}
\newcommand{\inn}[2]{\langle#1|#2\rangle}

\theoremstyle{plain}

\newtheorem{theorem}{Theorem}
\newtheorem{lemma}{Lemma}

\theoremstyle{definition}
\newtheorem{definition}{Definition}

\usepackage[table]{xcolor}
\usepackage{booktabs, array}
\definecolor{natureblue}{RGB}{235, 245, 255} % 아주 연한 푸른색 (강조용)
\definecolor{naturegrey}{RGB}{245, 245, 245} % 아주 연한 회색 (구분용)

\usepackage{quantikz}
\usepackage{xcolor}

\usepackage[most]{tcolorbox}

\newtcolorbox{protocolbox}[1]{
    colback=white, 
    colframe=gatedark!50, % 이전에 정의한 gatedark 색상 활용
    fonttitle=\bfseries\sffamily,
    coltitle=white,
    colbacktitle=gatedark!80,
    enhanced,
    attach boxed title to top left={yshift=-2mm, xshift=3mm},
    title=#1,
    arc=2mm, % 모서리 둥글게
    boxrule=0.8pt,
}

\definecolor{gateblue}{RGB}{210, 235, 250} % 게이트 배경색
\definecolor{gatedark}{RGB}{40, 100, 150} % 게이트 테두리 및 텍스트 색상

\tikzset{
    operator/.append style={fill=gateblue, draw=gatedark, thick},
    meter/.append style={fill=gray!10, draw=black},
    phase/.append style={fill=gatedark}
}

\begin{document}

\title{Near optimal quantum algorithm for estimating Shannon entropy}

\author{Myeongjin Shin\href{https://orcid.org/0000-0002-8246-6300}{\includegraphics[scale=0.05]{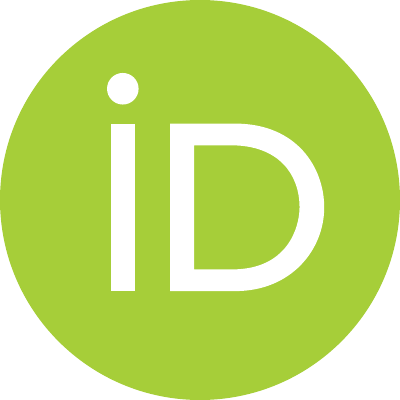}}
}
\email{hanwoolmj@kaist.ac.kr}
\affiliation{School of Computing, KAIST, Daejeon 34141, Korea}
\affiliation{Team QST, Seoul National University, Seoul 08826, Korea}

\author{Kabgyun Jeong\href{https://orcid.org/0000-0001-7628-7835}{\includegraphics[scale=0.05]{orcidid.pdf}}
}
\email{kgjeong6@snu.ac.kr}
\affiliation{Research Institute of Mathematics, Seoul National University, Seoul 08826, Korea}
\affiliation{School of Computational Sciences, Korea Institute for Advanced Study, Seoul 02455, Korea}
\affiliation{Team QST, Seoul National University, Seoul 08826, Korea}

\begin{abstract}
We present a near-optimal quantum algorithm, up to logarithmic factors, for estimating the Shannon entropy in the quantum probability oracle model. Our approach combines the singular value separation algorithm with quantum amplitude amplification, followed by the application of quantum singular value transformation. On the lower bound side, we construct probability distributions encoded via Hamming weights in the oracle, establishing a tight query lower bound up to logarithmic factors. Consequently, our results show that the tight query complexity for estimating the Shannon entropy within $\epsilon$-additive error is given by $\tilde{\Theta}\left(\tfrac{\sqrt{n}}{\epsilon}\right)$.
\end{abstract}

\maketitle

%\tableofcontents

%\section*{Key Points}
%\begin{itemize}
%\item AI models can be leveraged to represent and characterize scalable quantum systems in a data-driven manner for the tasks of quantum property prediction and implicit and approximate quantum state reconstruction. 
%\item Provably efficient machine learning models have been designed to characterize linear properties of scalable quantum systems and to classify quantum phases. 
%\item Deep learning models offer powerful tools for predicting a wide range of quantum properties through representation learning, as well as for implicitly reconstructing quantum states using generative modeling approaches.
%\item Language models, building on the GPT architecture, provide a flexible framework for auto-regressively representing large families of quantum states, paving the way toward foundation models for quantum systems and enabling new directions for research and application.
%\end{itemize}

%%%%%
\section{Introduction} \label{sec: intro}
\begin{table*}[t]
    \centering
    \caption{\textbf{Comparison of complexity bounds for Shannon entropy estimation.} The results are summarized based on the quantum probability oracle model. Bold rows indicate the contributions of this work.}
    \label{tab:summary}
    \small
    \renewcommand{\arraystretch}{1.4} % 행 간격을 살짝 넓혀 가독성 향상
    \begin{tabular}{@{} llll @{}}
        \toprule
        \textbf{Reference} & \textbf{Complexity} & \textbf{Type} & \textbf{Method} \\ 
        \midrule
        \rowcolor{natureblue} % 우리 논문 결과 강조
        \textbf{This paper} & $\mathcal{\tilde{O}}(\sqrt{n}/\epsilon)$ & Upper bound & Quantum singular value separation \& QSVT \\ 
        Ref.~\cite{li2018quantum} & $\mathcal{\tilde{O}}(\sqrt{n}/\epsilon^2)$ & Upper bound & Quantum Monte Carlo method \\ 
        Ref.~\cite{gilyen2019distributional} & $\mathcal{\tilde{O}}(\sqrt{n}/\epsilon^{1.5})$ & Upper bound & Quantum singular value transformation (QSVT) \\ 
        \midrule
        \rowcolor{natureblue} % 우리 논문 결과 강조
        \textbf{This paper} & $\Omega(\sqrt{n}/\epsilon)$ & Lower bound & Quantum oracle composition \& Polynomial method \\ 
        Ref.~\cite{bun2018polynomial} & $\Omega(\sqrt{n})$ & Lower bound & Polynomial method \\ 
        \bottomrule
    \end{tabular}
\end{table*}
Random processes arise in various scientific domains, including statistical physics, information theory, and machine learning, where they provide fundamental tools for analysis and prediction. Each event in a random process occurs with a certain probability. Accurately estimating these probabilities or computing information-theoretic quantities derived from them constitutes a fundamental challenge across both classical and quantum settings. Among such quantities, the Shannon entropy~\cite{bromiley2004shannon} plays a central role in characterizing randomness, quantifying information content, and finding applications in areas such as thermodynamics. Thus, efficient entropy estimation is of both theoretical and practical importance. In this work, we propose an efficient quantum algorithm for estimating the Shannon entropy within the quantum probability oracle model.

To investigate the potential advantage of quantum models for analyzing random processes, previous works have proposed various quantum frameworks~\cite{gilyen2019distributional, van2021quantum} and examined their benefits in comparison to the classical probability sampling model. Among these, the quantum probability oracle model has emerged as the most widely used and well-established framework.

\begin{definition}[Quantum probability oracle]\label{def:qoracle}
    Let $p$ be a $n$-dimensional probability distribution. We say that $O_p$ is a quantum probability oracle for $p$ if
    \begin{equation}
        \ket{\psi}_p = O_p\ket{0} = \sum_{i=1}^n \sqrt{p_i}\ket{i}\ket{\psi_i}
    \end{equation}
    for some orthogonal quantum states $\{\ket{\psi_i}\}_{i=1}^n$.
\end{definition}

This model is the quantum analogue of the classical sampling model, since applying the oracle $O_p$ to the state $\ket{0}$ and measuring the first register is equivalent to drawing a sample from the distribution $p$. We refer to applying a quantum state to the oracle $O_p$ as making a ``query" to the quantum oracle. A key distinguishing feature of the quantum model, compared to the classical one, is that it also allows queries to the inverse oracle $O_p^{\dagger}$ (as well as controlled operations), thereby enabling genuine quantum speedups and advantages. 

An notable example of such a speedup is provided by the quantum amplitude estimation technique. In the classical setting, estimating a probability $p_i$ to within an additive error $\epsilon$ with success probability $1-\delta$ requires $\Theta\big(\frac{\ln(1/\delta)}{\epsilon^2}\big)$ samples~\cite{dagum2000optimal}. In contrast, given access to a quantum probability oracle, the amplitude estimation algorithm~\cite{brassard2000quantum} reduces this to $\mathcal{O}\big(\frac{\ln(1/\delta)}{\epsilon}\big)$ applications of $O_p$ and its inverse $O_p^{\dagger}$. Recent work has further shown that the availability of the inverse oracle $O_p^{\dagger}$ is essential to achieve this quantum advantage~\cite{tang2025amplitude}.

In this paper, we establish the tight complexity bound (up to logarithmic factors) for estimating the Shannon entropy in the quantum probability oracle model. For a discrete distribution $p=(p_i)^n_{i=1}$ supported on $[n]$, Shannon entropy is defined as:
\begin{equation} \label{eq: shannon}
    H(p) = -\sum_{i=1}^n p_i\log p_i.
\end{equation}

The tight classical sample complexity for estimating $H(p)$ within additive error $\epsilon$ is $\mathcal{O}(\frac{n}{\epsilon\log n}+\frac{(\log n)^2}{\epsilon^2})$~\cite{wu2016minimax}. Quantum algorithms that exploit the power of the quantum probability oracle can surpass this classical bound. Initially, a quantum query complexity $\tilde{\mathcal{O}}(\frac{\sqrt{n}}{\epsilon^2})$ was established for $H(p)$ estimation~\cite{li2018quantum}, which uses the quantum Monte Carlo method~\cite{montanaro2015quantum}. Later, employing quantum singular value transformation (QSVT), this complexity was improved to $\tilde{\mathcal{O}}(\frac{\sqrt{n}}{\epsilon^{1.5}})$, which has since become well-known ``folklore" upper bound~\cite{gilyen2019distributional}. We present an algorithm that surpasses the folklore quantum query complexity by integrating singular value separation with quantum amplitude amplification. We acknowledge that our algorithm was inspired by the framework of variable time amplitude estimation, which was previously applied to R\'{e}nyi entropy estimation~\cite{wang2024quantum}.

The folklore complexity matches the known lower bound of $\Omega(\sqrt{n})$ in terms of $n$. Specifically, the lower bound for estimating $H(p)$ to within a constant additive error is $\Omega(\sqrt{n})$, established via the polynomial method~\cite{bun2018polynomial}. However, to the best of our knowledge, no lower bounds have been proven for estimation with arbitrarily small additive error $\epsilon$. In this work, we provide the first such lower bound for arbitrary $\epsilon$, thus closing this gap in the quantum query complexity of Shannon entropy estimation. Our main result is summarized in Theorem~\ref{thm: main}, which establishes the tight bound for estimating Shannon entropy.

\begin{theorem}[Main Theorem]\label{thm: main}
    Let $p$ be a $n$-dimensional probability distribution. Given a quantum probability oracle $O_p$ for $p$, estimating the Shannon entropy $H(p)$ within $\epsilon$-additive error involves $\tilde{\Theta}(\frac{\sqrt{n}}{\epsilon})$ queries of $O_p$ and $O_p^\dagger$.
\end{theorem}

Table~\ref{tab:summary} summarizes our results and compares them with previous work. Our study is the first to establish the tight query complexity for estimating the Shannon entropy in a quantum probability oracle. By leveraging the power of quantum singular value separation, we propose an improved upper bound. Furthermore, we prove a tight lower bound, up to logarithmic factors, by employing quantum oracle composition.

The remainder of the paper is organized as follows. In Section~\ref{sec: prelim} we review the quantum amplitude amplification and estimation method. And also review the quantum singular value transformation and it's application to the folklore Shannon entropy estimation algorithm and quantum singular value separation algorithm. Section~\ref{sec: upper} presents the quantum singular value separation algorithm, which forms a key component of our main algorithm, and provides a detailed complexity analysis. Section~\ref{sec: lower} establishes the lower bound for Shannon entropy estimation, which matches our algorithm's complexity. Section~\ref{sec: discussion} discusses future work and open problems.

\section{Preliminaries} \label{sec: prelim}
\subsection{Quantum amplitude amplification and estimation} \label{subsec: qaae} 
Quantum amplitude amplification and estimation generalize Grover’s search algorithm, providing a quadratic speedup over classical methods~\cite{brassard2000quantum}. The version presented in Lemma~\ref{lem: qaa} is also known as fixed-point quantum search with an optimal number of queries~\cite{yoder2014fixed}, which avoids the ``overcooking" problem. 

Suppose we are given a unitary operator $A$ that prepares the initial state $\ket{S}$ = $A\ket{0}^{\otimes n}$. From $\ket{S}$, we would like to extract the target state $\ket{T}$ with success probability $p_L\geq 1-\delta$, where the overlap $\inn{T}{S}=\sqrt{\lambda}e^{i\phi}$ with $\lambda\neq0$, and $\delta\in[0,1]$ is given. We are also provided with an oracle $U$ which flips an ancilla qubit when fed the target state: 
\begin{align}
    U\ket{T}\ket{b}=\ket{T}\ket{b\oplus1}, U\ket{\tilde{T}}\ket{b}=\ket{\tilde{T}}\ket{b} \quad \nonumber \\
    \text{for} \quad \inn{T}{\tilde{T}}=0.
\end{align}

\begin{lemma}[Quantum amplitude amplification]\label{lem: qaa}
    We can construct a unitary $V$ such that
    \begin{align}
        & V\ket{S}\ket{0}=\ket{T'}\ket{0}, \;\text{and}  \\
        & |\inn{T}{T'}|^2 \geq 1-\delta^2
    \end{align}
    using $L$ queries to $U, A, A^\dagger$, and efficiently implementable $n$-qubit gates, where
    \begin{equation}
        L=\mathcal{O}\left(\log\left(\frac{2}{\delta}\right)\frac{1}{\sqrt{\lambda}}\right).
    \end{equation}
\end{lemma}

Quantum amplitude amplification serves as a key component in our main algorithm. In particular, we can construct a unitary oracle $U$ that amplifies specific states of interest. Another central tool is the quantum amplitude estimation technique, described in Lemma~\ref{lem: qae}.

\bigskip
\begin{lemma}[Quantum amplitude estimation]\label{lem: qae}
    Let $p=\{p_i\}^n_{i=1}$ be an $n$-dimensional probability distribution and let $O_p$ be a quantum probability oracle for $p$. Quantum amplitude estimation outputs estimates $\tilde{p_i}\in[0,1]$ such that
    \begin{equation}
        |\tilde{p_i}-p_i| \leq \frac{2\pi\sqrt{p_i(1-p_i)}}{M} + \frac{\pi^2}{M^2},
    \end{equation}
    with success probability of at least $\frac{8}{\pi^2}$, using $M$ calls to $O_p$ and $O_p^\dagger$.
\end{lemma}
In particular, without prior knowledge of $p_i$, we can estimate $p_i$ within additive error $\epsilon$ using $\mathcal{O}(\frac{1}{\epsilon})$ queries to $O_p$ and $O_p^\dagger$. This provides a quadratic speedup over the classical sample complexity $\mathcal{O}(\frac{1}{\epsilon^2})$.

\subsection{Quantum singular value transformation (QSVT)} \label{subsec: qsvt}
Singular value transformation is one of the most powerful tools in quantum algorithms, enabling the application of a function $f$ to the eigenvalues (or singular values) of Hermitian operators. Formally it is defined as follows.
\begin{definition}[Singular value transformation]\label{def:PolySVTrans}
    Let $f:\mathbb{R}\rightarrow\mathbb{C}$ be an even or odd function. Let $A\in\mathbb{C}^{\tilde{d}\times d}$ have the following singular value decomposition
    \begin{equation*}                     
            A=\sum_{i=1}^{d_{\min}}\varsigma_i\ket{\tilde{\psi}_i}\bra{\psi_i},
    \end{equation*}
    where $d_{\min}:=\min(d,\tilde{d})$. For the function $f$       we define the singular value transformation on $A$ as
    \begin{equation*}
            f^{(SV)}(A):=\left\{\begin{array}{rcl} \sum_{i=1}^{d_{\min}}f(\varsigma_i)\ket{\tilde{\psi}_i}\bra{\psi_i}& &\text{if }f\text{ is odd, and}\\[\medskipamount]
            \sum_{i=1}^{d}f(\varsigma_i)\ket{\psi_i}\bra{\psi_i}& & \text{if }f\text{ is even,} \end{array}\right.
    \end{equation*}	
where for $i\in[d]\setminus[d_{\min}]$ we define $\varsigma_i:=0$.
\end{definition}
    
Quantum singular value transformation (QSVT) is the quantum analogue of the classical singular value transformation and has proven to be a powerful tool for property testing and related algorithmic tasks~\cite{gilyen2019quantum}. In particular, QSVT with real polynomial transformations can be implemented on a quantum computer.

\begin{lemma}[Ref \cite{gilyen2019quantum}, Corollary 18]\label{lem: QSVT-gilyen}
Let $\mathcal{H}_U$ be a finite-dimensional Hilbert space and let $U,\Pi, \widetilde{\Pi}\in\text{End}(\mathcal{H}_U)$ be linear operators on $\mathcal{H}_U$ such that $U$ is a unitary, and $\Pi, \widetilde{\Pi}$ are orthogonal projectors. Suppose that $P=\sum_{k=0}^{n}a_k x^k\in\mathbb{R}[x]$ is a degree-$n$ polynomial such that
\begin{itemize}
    \item $a_k\neq 0$ only if $\,k \equiv n \mod 2$, and
    \item for all $x\in[-1,1]\colon$ $| P(x)|\leq 1$.
\end{itemize}
Then we can apply singular value transformation of the matrix $\widetilde{\Pi} U \Pi$ as below
\begin{equation}
	P^{(SV)}\!\left(\widetilde{\Pi}U\Pi\right)\!
\end{equation}
with $n$ uses of $U$, $U^\dagger$ and the same number of controlled reflections $I\!-\!2\Pi$ and $I\!-\!2\widetilde{\Pi}$. 

\end{lemma}

Let us explore how we can apply QSVT to the quantum probability oracle model. Suppose that $S$ is the $k$-degree polynomial that approximates the function $f$, which is the function we want to apply the singular value transformation. Let the projection operators $\widetilde{\Pi},\Pi$ be
\begin{align}
    \widetilde{\Pi}&=\sum_{i=1}^n I\otimes\ket{i}\bra{i}\otimes\ket{i}\bra{i}, \;\text{and} \\
    \Pi&=\ket{0}\bra{0}\otimes\ket{0}\bra{0}\otimes I,
\end{align}
and let $P=S$ and $U=O_p$ in Definition~\ref{def:qoracle}. %We have
%\begin{equation}
%    \widetilde{\Pi}U\Pi = \sum_{i=1}^n \sqrt{p_i}\ket{i}\bra{0}\otimes\ket{\psi_i}\bra{0}\otimes\ket{i}\bra{i}
%\end{equation}

Then we can apply the singular value transformation in Lemma~\ref{lem: QSVT-gilyen}, which satisfies
\begin{equation}
S^{(SV)}\!\left(\widetilde{\Pi}U\Pi\right)\! 
=\sum_{i=1}^nS(\sqrt{p_i})\ket{i}\bra{0}\otimes\ket{\psi_i}\bra{0}\otimes\ket{i}\bra{i}.
\end{equation}

By exploiting QSVT with controlled projection operators, we can implement a unitary operator $U_S^{(SV)}$ that satisfies
\begin{align}\label{eq: QSVT-U-SV}
& (C_{\widetilde{\Pi}\otimes\ket{+}\bra{+}}NOT)(U_S^{(SV)}\otimes I)(\ket{0}\ket{0}(O_p\ket{0})\ket{+}\ket{0}) \\ \nonumber
&=\ket{\psi_{\text{garbage}}}\ket{0}+\sum_{i=1}^n\sqrt{p_i}S(\sqrt{p_i})\ket{i}\ket{\psi_i}\ket{i}\ket{\psi_i}\ket{+}\ket{1},
\end{align}
where $U_S^{(SV)}$ encompasses $k$ applications of $O_p, O_p^\dagger$. The elaborated proof is described in Lemma 2 of~\cite{wang2024quantum}.

The probability of last qubit becoming $\ket{1}$ in Equation~\ref{eq: QSVT-U-SV} is $\sum_{i=1}^np_iS(p_i)^2$, applying quantum amplitude estimation allows us to estimate the quantity within additive error $\epsilon$ using $\mathcal{O}(\frac{1}{\epsilon})$ queries to $U_S^{(SV)}$ and its inverse. Since $U_S^{(SV)}$ encompasses $k(=\text{deg}(S))$ applications of $O_p,O_p^\dagger$, estimation of $\sum_{i=1}^np_iS(p_i)^2$ within additive error $\epsilon$ are obtained by $\mathcal{O}(\frac{k}{\epsilon})$ queries.

To apply singular value transformation to our problem of estimating Shannon entropy, we need low-degree polynomial approximates to the following function $\sqrt{\log(\frac{1}{x})}$.

\begin{lemma}[Ref \cite{wang2024new}, Lemma 3.3]\label{lem:polyapx}
    Let $\beta\in(0,\frac{1}{2}]$, $\eta\in(0,\frac{1}{2}]$ and $t\geq 1$. There exist a polynomial $\tilde{S}$ such that
	\begin{itemize}
            \item $\forall x\in [\beta,1-\beta]\colon |\tilde{S}(x)-\frac{\sqrt{\log(1/x)}}{2\sqrt{\log(1/\beta)}}|\leq\eta$,	and $\,\forall x\in[-1,1]\colon -1\leq\tilde{S}(x)=\tilde{S}(-x)\leq 1$,
	\end{itemize}
	moreover $\deg(\tilde{S})=\mathcal{O}\left({\frac{1}{\beta}\log\left(\frac{1}{\beta\eta}\right)}\right)$.
\end{lemma}

By combining the lemmas above, Gilyén and Li~\cite{gilyen2019distributional} obtained the following result. 

\begin{lemma}\label{lem: QSVT-Shannon}
    Let $p=\{p_i\}^n_{i=1}$ be a $n$-dimensional probability distribution and $O_p$ is a quantum probability oracle for $p$. And 
    let $\beta$ be a threshold parameter. Then, $H(p)$ can be estimated within additive error  $(\epsilon+\sum_{p_i<\beta}p_i)$ using
    \begin{equation}
        \tilde{\mathcal{O}}\left(\frac{1}{\epsilon\sqrt{\beta}}\right)
    \end{equation}
    queries to $O_p$ and $O_p^\dagger$.
\end{lemma}
By setting $\epsilon=\frac{\epsilon}{2}$, $\beta=\frac{\epsilon}{2n}$ in Lemma~\ref{lem: QSVT-Shannon}, one can estimate $H(p)$ within additive error $\epsilon$ using $\mathcal{O}(\frac{\sqrt{n}}{\epsilon^{1.5}})$ queries to $O_p$ and $O_p^\dagger$, which is the folklore query complexity for Shannon entropy estimation.

We next introduce another useful method, singular value separation, which employs QSVT and serves as a key component of our algorithm. Using QSVT, one can decompose a quantum state into multiple components by separating singular values. 

\begin{lemma}[Ref \cite{wang2024quantum}, Lemma 5]\label{lem: sv-separation}
    Let $U$ be a unitary, and $\widetilde{\Pi},\Pi$ orthogonal projectors with the same rank $d$ acting on $\mathcal{H}_I$. Suppose $A=\widetilde{\Pi}U\Pi$ has a singular value decomposition $A=\sum_{i=1}^d\sigma_i\ket{\tilde{\psi}_i}\bra{\psi_i}_I$. Let $\varphi\in(0,1]$ and $\epsilon>0$. Then there is a unitary $W(\varphi,\epsilon)$ using $\mathcal{O}(\frac{1}{\varphi}\log\frac{1}{\epsilon})$ queries to $U,U^\dagger$ such that
    \begin{equation}
        W(\varphi,\epsilon)\ket{0}_C\ket{0}_P\ket{\psi_i}_I = \beta_0\ket{0}_C\ket{\gamma}_{P,I}+\beta_1\ket{1}_C\ket{+}_P\ket{\psi_i}_I,
    \end{equation}
    where $|\beta_0|^2+|\beta_1|^2=1$, such that
    \begin{itemize}
        \item if $0\leq\sigma_i\leq\varphi$, then $|\beta_1|\leq\epsilon$ and
        \item $2\varphi\leq\sigma_i\leq1$, then $|\beta_0|\leq\epsilon$.
    \end{itemize}
    Here $C$ and $P$ are single-qubit registers, and $I$ is the register on which $A$ acts. 
\end{lemma}

Upon applying $W(\varphi, \epsilon)$ to $\ket{0}_C\ket{0}_P\ket{\psi_i}_I$, the single-qubit register $C$ contains probabilistic information indicating whether the eigenvalue $\sigma_i$ lies below or above the threshold $\varphi$. In the next section, we explore the efficient application of Lemma~\ref{lem: sv-separation} to the quantum probability oracle setting. 

%%%%
\section{Upper bound}\label{sec: upper}

In this section, we present our main algorithm (Algorithm~\ref{alg:Shannon}) and analyze its complexity, thereby obtaining the upper bound for Shannon entropy estimation. As a key ingredient, we employ quantum singular value separation, first introduced in~\cite{wang2024quantum}. We describe the details in the following subsections.

\subsection{Quantum singular value separation algorithm}

In this subsection, we elaborate on the Quantum singular value separation algorithm (depicted as Figure~\ref{fig:qsvs_full}) and explain why it separates the singular values. Singular value separation partitions the singular values and their corresponding singular vectors into a quantum state that we can be accessed and manipulated. We divide the singular values into intervals 
\begin{equation}
    [\varphi_{m},\varphi_{m-1}),[\varphi_{m-1},\varphi_{m-2}),\dots,[\varphi_2,\varphi_1),[\varphi_1,\varphi_0],
\end{equation}
where $m=\mathcal{\tilde{O}}(\log\frac{\epsilon}{n})$ and $\varphi_j = \frac{1}{2^j}$. 

When utilizing QSVT with $O_p$, the singular values are $\sqrt{p_1},\sqrt{p_2},\dots,\sqrt{p_n}$. So our aim is dividing that singular values into the interval ``buckets". \\

\paragraph*{\textbf{Ancila Registers}.}

To apply the singular value separation to $\ket{\psi}_{AB}=O_p\ket{0}_{AB}=\sum_{i=1}^n\sqrt{p_i}\ket{i}_A\ket{\psi_i}_B$, we prepare the registers 

\begin{align}
    C&=(C_1,C_2,\dots,C_m), P=(P_1,P_2,\dots,P_m),\quad\text{and} \nonumber\\
    I&=(I_1,I_2,\dots,I_m),
\end{align}
where $C_i$ and $P_i$ are single-qubit registers and each $I_i$ consists of $\lceil3\log n\rceil$ qubits for $i=1,2,\dots,m$.

To facilitate the description of our multi-register quantum algorithm, we define the following register structure and state configurations: \\

1. Register Partitioning \\
We partition the auxiliary space into three primary groups $C, P$, and $I$. For any range $1 \le x < y \le m$, we denote the sub-registers as:
\begin{align*}
    C_{x,y} &= (C_x, C_{x+1}, \dots, C_y), \\
    P_{x,y} &= (P_x, P_{x+1}, \dots, P_y), \\
    I_{x,y} &= (I_x, I_{x+1}, \dots, I_y),
\end{align*}
where $\dim(C_i) = \dim(P_i) = 2$ and $\dim(I_i) = 2^{\lceil 3\log n \rceil}$. \\

2. Target State Definitions \\
The indicator state $\ket{\lambda_j}_C$ represents a configuration where only the $j$-th control qubit is active:
\begin{equation*}
    \ket{\lambda_j}_C = \underbrace{\ket{0 \dots 0}_{C_{1,j-1}}}_{\text{all zero}} \otimes \underbrace{\ket{1}_{C_j}}_{\text{j-th}} \otimes \underbrace{\ket{0 \dots 0}_{C_{j+1,m}}}_{\text{rest}}
\end{equation*}
For the null case, we define the all-zero state as
\begin{equation*}
    \ket{\lambda_0}_C = \ket{0}_{C_1} \ket{0}_{C_2} \dots \ket{0}_{C_m} = \ket{0}_C
\end{equation*}
valid for $j \in \{1, 2, \dots, m\}$. \\

\noindent These notations allow for a concise representation of the coherent singular value separation across $m$ parallel branches. \\

\paragraph*{\textbf{Definition of useful unitaries}}
We utilize the unitary $W$ (defined in Lemma~\ref{lem: sv-separation}) on $C_j, P_j, I_j$ registers (for each $j$) with specified angles $\varphi_j$. 

\begin{protocolbox}{Unitary: $W$ Gate}
    \centering
    % 회로도 최적화 (측정 결과 포함)
    \begin{quantikz}[row sep={0.7cm, between origins}, column sep=0.55cm]
        \lstick{$\ket{0}_{C_j}$} & \gate[3]{W(\varphi_j,\epsilon)} & \meter{} \setwiretype{c} \ar[r] & \small{\begin{cases} \ket{0} \; (\sqrt{p_i} \le \varphi_j) \\ \ket{1} \;(\sqrt{p_i} \ge 2\varphi_j) \end{cases}} \\
        \lstick{$\ket{0}_{P_j}$} & \qw & \qw \\
        \lstick{$\ket{0,0,i}_{I_j}$} & \qw & \qw
    \end{quantikz}

    \vspace{4mm}
    \begin{minipage}{0.95\linewidth}
        \small
        \textbf{Action of $W(\varphi_j, \epsilon)$:}
        \begin{equation*}
            \ket{0,0,0,i}_{C,P,I} \xrightarrow{W} \beta_j' \ket{0}_{C_j} \ket{\gamma} + \beta_j \ket{1}_{C_j} \ket{+,0,0,i}
        \end{equation*}
        
        \textbf{Threshold Conditions:}
        \begin{itemize}
            \item \textbf{Small SV} ($0 \le \sqrt{p_i} \le \varphi_j$): $\beta_j^2 \le \epsilon^2$ (Close to $\ket{0}$)
            \item \textbf{Large SV} ($\sqrt{p_i} \ge 2\varphi_j$): $\beta_j^2 \ge 1-\epsilon^2$ (Close to $\ket{1}$)
        \end{itemize}
        \textit{Note: $\beta_j, \beta_j'$ are derived from the polynomial approximation of the step function.}
    \end{minipage}
\end{protocolbox}

\begin{lemma} \label{lem:W_operator_formal}
The operator $W(\varphi_j, \epsilon)$ satisfies the threshold conditions defined in the Box above. Detailed success probabilities and error bounds are provided in Appendix~\ref{appendix: W-operator}.
\end{lemma}

As our algorithm scales with multiple registers, we consolidate the definition of the state-copying operation $T_j$ as below:

\begin{protocolbox}{Unitary: $T_j$ Gate}
    \centering
    % 회로도 (박스 크기에 맞춰 약간 축소)
    \begin{quantikz}[row sep={0.6cm,between origins}, column sep=0.8cm]
        \lstick{$\ket{0,0,0}_{I_j}$} & \gate[2]{T_j} & \rstick{$\ket{0,0,i}_{I_j}$} \qw \\
        \lstick{$\ket{i}_{A}$} & & \rstick{$\ket{i}_{A}$} \qw
    \end{quantikz}
    
    \vspace{3mm}
    \begin{minipage}{0.95\linewidth}
        \textbf{Definition.} The $T_j$ gate is a unitary operator mapping indices from the ancilla to the index register:
        \begin{equation*}
            T_j \left( \ket{0,0,0}_{I_j} \ket{i}_A \right) = \ket{0,0,i}_{I_j} \ket{i}_A
        \end{equation*}
        where $j \in \{1, \dots, m\}$ denotes the $j$-th copy register.
    \end{minipage}
\end{protocolbox}

\noindent \textit{Note: The $T_j$ gate ensures that each index $i$ is consistently distributed across parallel registers for the subsequent $W$ gate operations.}

We now define the controlled version of the operator $W$ and $T_j$, which are the key components in our quantum singular value separation algorithm.

\begin{protocolbox}{Unitaries: Controlled $T_j$ and $W_j$}
    \small
    \textbf{1. Logical Definitions} \\
    The multi-controlled operators $CT_j$ and $CW_j$ are active only when all preceding control qubits $C_{1,j-1}$ are in the $\ket{0}$ state:
    \begin{align*}
        CT_j &= \ket{0}\bra{0}^{\otimes (j-1)}_{C_{1,j-1}} \otimes T_j + \left( I - \ket{0}\bra{0}^{\otimes (j-1)}_{C_{1,j-1}} \right) \otimes I \\
        CW_j &= \ket{0}\bra{0}^{\otimes (j-1)}_{C_{1,j-1}} \otimes W_j + \left( I - \ket{0}\bra{0}^{\otimes (j-1)}_{C_{1,j-1}} \right) \otimes I
    \end{align*}

    \vspace{2mm}
    \textbf{2. Circuit Decomposition} \\
    The following circuits implement these operations using multi-controlled structures:

    \centering
    \begin{tikzpicture}
        % CT_j Decomposition
        \node (ct) at (0,0) {
            \begin{quantikz}[row sep={0.55cm,between origins}, column sep=0.6cm]
                \lstick{$C$} & \ctrl{1} & \qw \\
                \lstick{$I_j, A$} & \gate{T_j} & \qw 
            \end{quantikz}
        };
        \node at (1.8, -0.3) {$\equiv$};
        \node (ct_decomp) at (4.2, 0) {
            \begin{quantikz}[row sep={0.45cm,between origins}, column sep=0.5cm]
                \lstick{$C_1$} & \ctrl{4} & \qw \\
                \lstick{$C_2$} & \ctrl{3} & \qw \\[-0.2cm]
                \lstick{$\vdots$} & \vdots & \qw \\[-0.2cm]
                \lstick{$C_{j-1}$} & \ctrl{1} & \qw \\
                \lstick{$I_j, A$} & \gate{T_j} & \qw 
            \end{quantikz}
        };
        \node[below=1.2cm, font=\bfseries\sffamily\footnotesize] at (2.1,0) {(a) $CT_j$ Decomposition};
    \end{tikzpicture}

    \vspace{4mm}

    \begin{tikzpicture}
        % CW_j Decomposition
        \node (cw) at (0,0) {
            \begin{quantikz}[row sep={0.55cm,between origins}, column sep=0.6cm]
                \lstick{$C$} & \ctrl{1} & \qw \\
                \lstick{$P_j, I_j$} & \gate{W_j} & \qw 
            \end{quantikz}
        };
        \node at (1.8, -0.3) {$\equiv$};
        \node (cw_decomp) at (4.5, 0) {
            \begin{quantikz}[row sep={0.45cm,between origins}, column sep=0.5cm]
                \lstick{$C_1$} & \ctrl{4} & \qw \\
                \lstick{$C_2$} & \ctrl{3} & \qw \\[-0.2cm]
                \lstick{$\vdots$} & \vdots & \qw \\[-0.2cm]
                \lstick{$C_{j-1}$} & \ctrl{1} & \qw \\
                \lstick{$C_j, P_j, I_j$} & \gate{W_j} & \qw 
            \end{quantikz}
        };
        \node[below=1.2cm, font=\bfseries\sffamily\footnotesize] at (2.1,0) {(b) $CW_j$ Decomposition};
    \end{tikzpicture}

    We will use the simplified version (the left) of the circuit.
\end{protocolbox} 

\paragraph*{\textbf{Quantum singular value separation algorithm.}} Using these unitaries as tools, we state the singular value separation algorithm and analyze its behavior in Theorem~\ref{thm:sv-separation}.

\begin{figure*}[t]
    \centering
    \begin{tikzpicture}
    % 1. Abstract Block (왼쪽)
    \node (left) at (0,0) {
        \begin{quantikz}[row sep={0.55cm,between origins}, column sep=0.8cm]
            \lstick{$C$} & \gate[5]{\text{QSVS}_k} & \qw \\
            \lstick{$P$} & \qw & \qw \\
            \lstick{$I$} & \qw & \qw \\
            \lstick{$A$} & \qw & \qw \\
            \lstick{$B$} & \qw & \qw
        \end{quantikz}
    };

    \node at (2.0, 0) {\large $\equiv$};

    % 2. Decomposition Block (오른쪽 통합)
    \node (main) at (9.0, 0) {
        \begin{quantikz}[row sep={0.55cm,between origins}, column sep=0.55cm]
            \lstick{$C$} & \ctrl{2} & \ctrl{1} & \ctrl{2} & \ctrl{1} & \qw & \dots & \ctrl{2} & \ctrl{1} & \qw \\
            \lstick{$P$} & \qw & \gate[2]{W_1} & \qw & \gate[2]{W_2} & \qw & \dots & \qw & \gate[2]{W_k} & \qw \\
            \lstick{$I$} & \gate[2]{T_1} & \qw & \gate[2]{T_2} & \qw & \qw & \dots & \gate[2]{T_k} & \qw & \qw \\
            \lstick{$A$} & \qw & \qw & \qw & \qw & \qw & \dots & \qw & \qw & \qw \\
            \lstick{$B$} & \qw & \qw & \qw & \qw & \qw & \dots & \qw & \qw & \qw
        \end{quantikz}
    };

    \end{tikzpicture}
    
    \vspace{4pt}
    \caption{\textbf{Circuit architecture for the iterative QSVS algorithm.} 
    The abstract $\text{QSVS}_k$ operation is decomposed into a sequence of $k$ cycles. 
    Each cycle consists of a state copying gate $T_j$ followed by a singular value separation gate $W_j$, both conditioned on the control register $C$. 
    Horizontal wires represent the coherent evolution across the control ($C$), phase ($P$), index ($I$), and system ($A, B$) registers.}
    \label{fig:qsvs_full}
\end{figure*}

The Quantum singular value separation algorithm utilizes the unitary defined as
\begin{equation}\label{eq: Uk}
    \text{QSVS}_k(\epsilon) = \prod_{j=1}^k C_jW(\varphi_j,\epsilon)C_jT.
\end{equation}
We demonstrate the power of $\text{QSVS}_k$ by applying it to $\ket{\psi}_{AB}$. Lemma~\ref{lem: mag}, which is deduced from Theorem~\ref{thm:sv-separation}, demonstrates that $\text{QSVS}_k$ effectively separates singular values. 

\begin{theorem}[Singular value separation algorithm]\label{thm:sv-separation}
    Let $\ket{\psi}_{AB}=O_p\ket{0}_{AB}$. Also, define $\ket{\Psi_k}$ as
    \begin{equation}\label{eq: sv}
        \ket{\Psi_k}=\text{QSVS}_k(\epsilon)(\ket{0}_C\ket{0}_P\ket{0}_I\ket{\psi}_{AB})
    \end{equation}
    Then,
    \begin{equation}
        |(\bra{i}_A\bra{\lambda_j}_C\otimes I)\ket{\Psi_k}|^2=p_iB_j(\sqrt{p_i})^2,
    \end{equation}
    for $j=1,2,\dots,k$ and
    \begin{equation}
        |(\bra{i}_A\bra{0}_C\otimes I)\ket{\Psi_k}|^2=p_i(B^{'}_k(\sqrt{p_i}))^2,
    \end{equation}
    where
    \begin{align}
        B_j^{'}(x) &= \prod_{i=1}^{j}\beta_i^{'}(x) \\
        B_j(x) &= B_{j-1}^{'}(x)\beta_j(x),
    \end{align}
    and $\beta_j,\beta_j'$ are derived from the polynomial approximation of the step function of the $W$ gate.
\end{theorem}
\begin{proof}
    See Appendix~\ref{appendix: sv-separation}.
\end{proof}

From Theorem~\ref{thm:sv-separation}, the output of applying $\text{QSVS}_k$ to $\ket{\psi}_{AB}$ distributes the information of $\ket{\psi}_{AB}$ into the register $C$. For each $i$, $\sqrt{p_i}\ket{i}$ are distributed into $C_j$ by proportion $B_j(\sqrt{p_i})^2$.

\begin{lemma}\label{lem: mag}
    Suppose $x\in[\varphi_j,\varphi_{j-1})$. Then,
    \begin{equation}
        B_j(x)^2+B_{j+1}(x)^2\geq 1-\mathcal{O}(j\epsilon^2).
    \end{equation}
\end{lemma}
\begin{proof}
    See Appendix~\ref{appendix: mag}.
\end{proof}
Note that $B_k'(\sqrt{p_i})^2+\sum_{j=1}^k B_j(\sqrt{p_i})^2 = 1$. $B_j(\sqrt{p_j})^2, B_{j+1}(\sqrt{p_j})^2$ occupies most of the portion. Thus, Lemma~\ref{lem: mag} implies that most of the information corresponding to a singular value $\sqrt{p_i}\in[\phi_j,\phi_{j-1})$ is concentrated in registers $C_j$ and $C_{j+1}$. 

Since the information of the singular value $p_i$ is concentrated in few registers ($C_j, C_{j+1}$), we can use quantum amplitude amplification to amplify the information in those registers and apply the quantum singular value transformation. This allows us to achieve an an efficient (Heisenberg limit) query complexity compared to previous algorithms~\cite{li2018quantum, gilyen2019distributional}.

The following lemma states the complexity of running the quantum singular value separation algorithm, which is equivalent to the query complexity of $\text{QSVS}_k(\epsilon)$.

\begin{lemma}\label{lem: Uke}
    We can query to $\text{QSVS}_k(\epsilon)$ by using
    \begin{equation}
        \mathcal{O}\left(2^k\log\frac{1}{\epsilon}\right) 
    \end{equation}
    queries to $O_p$ and $O_p^\dagger$. 
\end{lemma}
\begin{proof}
    See Appendix~\ref{appendix: Uke}.
\end{proof}

\begin{widetext}
\begin{algorithm}[t]
\caption{Quantum algorithm for estimating Shannon entropy $H(p)$}
\label{alg:Shannon}
\KwIn{Quantum probability oracle $O_p$ and its inverse $O_p^\dagger$}
\KwOut{$H(p)$}

Prepare registers $C=(C_1,C_2,\dots,C_m),P=(P_1,P_2,\dots,P_m),I=(I_1,I_2,\dots,I_m)$ and $F$, where $C_i,P_i,F$ are single qubits and each $I_i$ has $\lceil\log n\rceil$ qubits.

Define polynomials $S_k$ for $k=1,2,\dots,m$ \tcp*{Definition~\ref{def: poly-S}}

\For{$k=1,2,\dots,m$}{
    Prepare $\ket{\Psi_k}$ \tcp*{Algorithm~\ref{thm:sv-separation}}
    
    Apply quantum amplitude amplification to obtain $\ket{\phi_k}$ \tcp*{Theorem~\ref{thm: pick}}

    Apply QSVT to obtain $U_{S_k}^{(SV)}$ \tcp*{Lemma~\ref{lem: QSVT-gilyen}}

    Use $(C_{\widetilde{\Pi}\otimes\ket{+}\bra{+}}NOT)(U_{S_k}^{(SV)}\otimes I)$ on $\ket{\phi_k}$ with some auxiliary qubits and apply quantum amplitude estimation to obtain $v_k^{'}=\frac{1}{\text{Sum}(k)}\sum_{i=1}^n p_iS_k(\sqrt{p_i})^2B_k(\sqrt{p_i})^2\log\frac{1}{\varphi_k}$ \tcp*{Lemma~\ref{lem: qaa}}

    Apply quantum amplitude estimation on $\ket{\Psi_k}$ to obtain $\text{Sum}(k)$ \tcp*{Lemma~\ref{lem: qaa}}

    Multiply $v_k^{'}\times\text{Sum}(k)$ to obtain $v_k$ \tcp*{Theorem~\ref{thm: v}}
}

Calculate $v=\sum_{k=1}^m 8v_k-2$ \tcp*{Theorem~\ref{thm: v}}

\Return{$v$}
\end{algorithm}
\end{widetext}

%%%%

\begin{figure*}[t]
    \centering
    \resizebox{0.95\textwidth}{!}{
        \begin{tikzpicture}[
            node distance=1.1cm and 1.3cm, 
            block/.style={rectangle, draw, fill=blue!5, text width=2.2cm, align=center, minimum height=1.1cm, font=\normalsize\sffamily, rounded corners},
            sum/.style={circle, draw, fill=orange!10, minimum size=0.9cm, inner sep=0pt, font=\normalsize},
            arrow/.style={-{Stealth[scale=1.2]}, thick},
        ]

        % --- Upper Path ---
        \node (input1) {$\ket{\psi}_{AB}$};
        
        \node [block, right=of input1] (qsvs) {QSVS \\ ($QSVS_k(\epsilon)$)};
        \node [block, right=1.0cm of qsvs] (qaa) {QAA};
        \node [block, right=1.0cm of qaa] (qsvt) {QSVT \\ ($U_{S_k}^{(SV)}$)};
        \node [block, right=1.0cm of qsvt] (qae1) {QAE};
        \node [right=0.8cm of qae1] (vk_prime) {$v_k'$};

        % --- Lower Path ---
        \node [block, below=1.8cm of qaa] (qae2) {QAE};
        \node [right=3.5cm of qae2] (sumk) {$\text{Sum}(k)$};

        % --- Multiplication/Final Node ---
        \node [sum, right=1.5cm of vk_prime] (mult) {$\times$};
        \node [right=0.8cm of mult] (final) {$v_k$};

        % --- Connections (Upper) ---
        \draw [arrow] (input1) -- (qsvs);
        \draw [arrow] (qsvs) -- (qaa);
        \draw [arrow] (qaa) -- (qsvt);
        \draw [arrow] (qsvt) -- (qae1);
        \draw [arrow] (qae1) -- (vk_prime);

        % --- Connections (Lower) ---
        \coordinate (branch) at ($(input1)!0.5!(qsvs)$);
        \draw [thick] (qsvs) |- (qae2);
        \draw [arrow] (qae2) -- (sumk);

        % --- Final Connections ---
        \draw [arrow] (vk_prime) -- (mult);
        \draw [arrow] (sumk.east) -| (mult.south);
        \draw [arrow] (mult) -- (final);

        \end{tikzpicture}
    }
    \caption{\textbf{Schematic workflow of Algorithm~\ref{alg:Shannon} for Shannon entropy estimation.} The upper path computes $v_k'$ by sequentially applying QSVS (Theorem~\ref{thm:sv-separation}), QAA (Lemma~\ref{lem: qaa}), QSVT (Lemma~\ref{lem: QSVT-gilyen}), and QAE (Lemma~\ref{lem: qae}) to the input state $\ket{\psi}_{AB}$. The lower path estimates $\text{Sum}(k)$ through QSVS and QAE. The product of these two outputs yields $v_k$, which is finally used to approximate the Shannon entropy as $H(p) \approx -2 + \sum_{k=1}^m 8v_k$. }
\end{figure*}

\subsection{Main algorithm}

Our main algorithm applies the singular value separation algorithm, followed by quantum amplitude amplification to extract desired states. QSVT and quantum amplitude estimation are then employed to estimate the Shannon entropy efficiently. To efficiently apply QSVT to the result of singular value separation algorithm, we define the following polynomials.

\bigskip
\begin{definition}\label{def: poly-S}
There exists a polynomial $S$ satisfying
    \begin{itemize}
        \item $\forall x\in [\varphi_{k},1]\colon \left|S(x)-\frac{\sqrt{\log(2/x)}}{2\sqrt{\log(1/\varphi_{k+1})}}\right|\leq\eta$ \;and \\ $\,\forall x\in[-1,1]\colon -1\leq\tilde{S}(x)=\tilde{S}(-x)\leq 1$, and
        \item $\deg(S)=\mathcal{O}\left({\frac{1}{\varphi_{k+1}}\log\left(\frac{1}{\eta\varphi_{k+1}}\right)}\right)$
    \end{itemize}
for $k=1,2,\dots,m$ by Lemma~\ref{lem:polyapx}. For each $k$, we denote the polynomial as $S_k$.
\end{definition}

Now we construct an approximate representation of $H(p)$ using $B_k$ and $S_k$.
\begin{theorem}\label{thm: v}
    Suppose $B_k$ is derived from $U_k(\delta)$ and $S_k$ is defined as above. Then
    \begin{align}
        v_k&=\sum_{i=1}^n p_iS_k(\sqrt{p_i})^2B_k(\sqrt{p_i})^2\log\frac{1}{\varphi_{k+1}},\;\text{and} \\
        v&=-2+\sum_{k=1}^m 8v_k
    \end{align}
    satisfies
    \begin{equation}
        |v-H(p)|=\mathcal{\tilde{O}}\left(m\delta^2+\eta+\frac{n}{2^m}\right).
    \end{equation}
\end{theorem}
\begin{proof}
    The proof is elaborated in Appendix~\ref{appendix: v}.
\end{proof}

By choosing parameters $\delta=\sqrt{\frac{\epsilon}{4m}}$, $\eta=\frac{\epsilon}{4}$ and $m=\log\frac{\epsilon}{2n}$, we obtain $|v-H(p)|\leq\mathcal{\tilde{O}}(\epsilon)$. Since $\delta$ and $\eta$ only contributes to the logarithmic terms of the complexity and $m$ is logarithmic to $\epsilon,n$, the complexity of estimating $v$ within additive error $\epsilon$ matches that of estimating $H(p)$, up to logarithmic factors.

Quantum amplitude amplification can be used to extract states associated with specific states. Now we examine how amplitude amplification is applied to extract certain states after the singular value separation algorithm. Applying some fundamental gates and Lemma~\ref{lem: qaa}, we can prove the following theorem.

\begin{theorem} \label{thm: pick}
    We define the quantum state $\ket{\phi_k}$ as
    \begin{align} \label{eq: pick}
    \ket{\phi_k} &= \ket{1}_{C_k}\sum_{i=1}^n\frac{\sqrt{p_i}B_k(\sqrt{p_i})}{\sqrt{\mathrm{Sum}(k)}}\ket{i}_A\ket{\psi_i}_B\ket{\text{garbage}}, \\
    \mathrm{Sum}(k) &= \sum_{i=1}^n p_iB_k(\sqrt{p_i})^2. 
    \end{align}
    There exists a quantum unitary $V$ satisfying
    \begin{align}
    & V \ket{0} = \ket{\phi'}\ket{0}, \\
    & \big|\inn{\phi'}{\phi_k}\big|^2 \geq 1 - \delta^2,
    \end{align}
    where $V$ can be implemented using
    \begin{equation}
        \mathcal{\tilde{O}}\left(\log\left(\frac{1}{\delta}\right)\frac{2^k}{\sqrt{\mathrm{Sum}(k)}}\right)
    \end{equation}
    queries to $O_p$ and $O_p^\dagger$, along with efficiently implementable elementary gates.

    % \begin{equation}
    % \mathcal{O}\left(\log(\frac{1}{\delta})\frac{1}{\sqrt{\mathrm{Sum}(k)}}\right)
    % \end{equation}
    % queries to $U_k(\epsilon)$, $U_k(\epsilon)^\dagger$, 
\end{theorem}
\begin{proof}
    The proof is elaborated in Appendix~\ref{appendix: pick}.
\end{proof}

Theorem~\ref{thm: pick} implies that states associated with $\ket{1}_{C_k}$ can be extracted. Since the amplification error $\delta$ only contributes a logarithmic term to the complexity, we ignore the effect of $\delta$. 

To estimate the value $v_k$, we employ the unitary $U_{S_k}^{(SV)},C_{\widetilde{\Pi}\otimes\ket{+}\bra{+}}NOT$ (in Equation~\ref{eq: QSVT-U-SV}) to $\ket{\phi_k}$ (in Equation~\ref{eq: pick}):

\begin{align}\label{eq: QSVT-amp-sv}
    &(C_{\widetilde{\Pi}\otimes\ket{+}\bra{+}}NOT)(U_{S_k}^{(SV)}\otimes I)(\ket{0}\ket{0}\ket{\phi_k}\ket{+}\ket{0}_F) \nonumber \\ &=\ket{1}_{C_k}\sum_{i=1}^n\frac{\sqrt{p_i}S_k(\sqrt{p_i})B_k(\sqrt{p_i})}{\sqrt{\text{Sum}(k)}}\ket{i}_A\ket{\psi_i}_B\ket{\text{garbage}_1}\ket{1}_F \nonumber \\ 
    &+\ket{1}_{C_k}\ket{\text{garbage}_2}\ket{0}_F.
\end{align}

The probability of measuring Equation~\ref{eq: QSVT-amp-SV} with $\ket{1}_{C_k}\ket{1}_F$ is $\sum_{i=1}^n p_iS_k(p_i)^2B_k(p_i)^2$, multiplying $\log\frac{1}{\varphi_{k+1}}$ we obtain $v_k$.

\begin{theorem}[Upper bound]\label{thm: upper}
    Let $p$ be a $n$-dimensional probability distribution. Given a quantum probability oracle $O_p$ for $p$, Algorithm~\ref{alg:Shannon} estimates the Shannon entropy $H(p)$ within $\epsilon$-additive error using 
    \begin{equation}
        \tilde{\mathcal{O}}\left(\frac{\sqrt{n}}{\epsilon}\right)
    \end{equation}
    queries of $O_p$ and $O_p^\dagger$.    
\end{theorem}
\begin{proof}
    The proof is elaborated in Appendix~\ref{appendix: upper}.
\end{proof}

Our algorithm is the first algorithm to estimate the Shannon entropy up to the Heisenberg limit and achieve square root dependency in terms of probability distribution size $n$.

\section{Lower bound}\label{sec: lower}
In this section, we prove that the upper bound established in Section~\ref{sec: upper} is essentially tight. Specifically, we show that any quantum algorithm estimating the Shannon entropy within additive error $\epsilon$ requires at least
\begin{equation}
    \Omega\left(\frac{\sqrt{n}}{\epsilon}\right)
\end{equation}
queries to the probability oracle $O_p$ and $O_p^\dagger$.

\begin{definition}[Classical distribution with discrete query-access]\label{def: doracle}
    A classical distribution $(p_i)^n_{i=1}$, has discrete query-access if we have classical / quantum query-access to a function $f: S\rightarrow [n]$ such that for all $i\in [n]$, $p_i = |{s\in[S] : f(s) = i}|/S$. In the quantum case a query oracle is a unitary operator $O$ acting on $\mathbb{C}^{|S|}\otimes\mathbb{C}^n$ as
    \begin{equation}
        O: \ket{s,0} \leftrightarrow \ket{s, f(s)} \quad \text{for all} \; s\in S.
    \end{equation}
\end{definition}
Note that if one first creates a uniform superposition over $S$ and then makes a query, then the above oracle turns into a quantum probability oracle as in Definition~\ref{def:qoracle}. Therefore, all lower bounds that are proven in this model also apply to the quantum probability oracle~\cite{gilyen2019distributional}. Lemma~\ref{lem: Hamming} proves the lower bound for obtaining the Hamming weight from a quantum oracle in Definition~\ref{def: doracle}. Lemma~\ref{lem: Shannon} proves the lower bound for estimating Shannon entropy within a constant error from a quantum oracle in Definition~\ref{def: doracle}.

\begin{lemma}\label{lem: Hamming}
    Let $x\in\{0,1\}^k$. Finding the Hamming weight $|x|$ requires $\Omega(k)$ quantum queries to a standard (binary) oracle for $x$~\cite{beals2001quantum}.
\end{lemma}

\begin{lemma}[Corollary 74 of~\cite{bun2018polynomial}]\label{lem: Shannon}
    Let $R = t \cdot n$ for a sufficiently large constant t. Interpret an input in $[R]^n$ as a distribution $p$ in the natural way (i.e., for each $j \in [n]$, $p_j = f_j/R$, where $f_j$ is the number of times $j$ appears in the input). There is a constant $c>0$  such that any quantum algorithm that approximates the entropy of p up to additive error $c$ with probability at least $2/3$ requires $\Omega(\sqrt{n})$ queries.
\end{lemma}

We combine Lemma~\ref{lem: Hamming} and Lemma~\ref{lem: Shannon} for the lower bound of estimating Shannon entropy within a desired error $\epsilon$. We design a quantum oracle where its probabilities are the Hamming weight of an different quantum oracle.

\begin{theorem}
    Let $\epsilon>0$. Any algorithm that (with success probability at least $\frac{2}{3}$) for every $n$-dimensional probability distribution $p$ outputs $H(p)$ within $\epsilon$-additive error, using queries to a quantum probability oracle for $p$, uses at least $\Omega(\frac{\sqrt{n}}{\epsilon})$ such queries.
\end{theorem}
\begin{proof}
    We acknowledge that the proof is similarly constructed  as Lemma 11 of~\cite{van2021quantum}, which proves the $\ell_1$-norm estimation lower bound.  

    Let 
    \begin{equation}
        k=\Theta\left(\frac{1}{\epsilon}\right)
    \end{equation}
    and
    \begin{equation}
        x^{(1)},\dots,x^{(n)}\in\{0,1\}^k, \quad x=\{x^{(1)}, x^{(2)},\dots,x^{(n)}\},
    \end{equation}
     and $t$ be a known constant such that $R=\sum_i|x^{(i)}|=tn$, where $|x^{(i)}|$ is the Hamming weight of $x^{(i)}$. Define
     \begin{equation}
         f_i=|x^{(i)}|,\;\; p_i=\frac{f_j}{R}
     \end{equation}
     as in Lemma~\ref{lem: Shannon}. We will explore the problem of estimating $H(p)=-p_i\log p_i$ with constant error $c$ in Lemma~\ref{lem: Shannon}.
    
    To estimate $H(p)$, we should retrieve $f_i$, in order to access the probability $p_i$. By Lemma~\ref{lem: Hamming}, finding the Hamming weight $f_i$ (or accessing a quantum analogue) requires $\Omega(k)$ queries. We further note that any algorithm that estimates $H(p)$ with constant error $c$ requires $\Omega(\sqrt{n})$ queries using Lemma~\ref{lem: Shannon}. Since quantum query complexity is multiplicative under composition~\cite{kimmel2012quantum} it follows
    that estimating $H(p)$ with constant error $c$, requires
    \begin{equation}
        \Omega(\sqrt{n}k)
    \end{equation}
    queries to $x$.

    Now we construct a slightly different quantum oracle using $x$. Define 
    \begin{equation}
        q=\{q_i\}, \;q_i=\frac{f_i}{nk}
    \end{equation}
    for $i\leq n$ and $q_{n+1}=1-\frac{R}{nk}=1-\frac{t}{k}$. We can sample from $q$ using a classical algorithm.
    \begin{enumerate}
        \item Pick a uniformly random $i\in[n]$.
        \item Pick a uniformly random $j\in[k]$.
        \item If $x_j^{(i)}=1$ return $i$, if $x_j^{(i)}=0$ return $n+1$.
    \end{enumerate}
    By replacing the uniformly random picks by the creation of a uniform superposition we get a
    quantum probability oracle for $q$. Now let us calculate $H(q)$ as follows

    \begin{align}
        H(q) &= -\sum_j \frac{f_j}{nk}\log\frac{f_j}{nk}-\left(1-\frac{t}{k}\right)\log\left(1-\frac{t}{k}\right) \nonumber \\
        &= \frac{t}{k}H(p)+B\left(\frac{t}{k}\right).
    \end{align}

    Since we know the constant $t$, we can retrieve $H(p)$ from $H(q)$ using the relation below, i.e,
    \begin{equation}\label{eq: shannon-relation}
        H(p) = \frac{k}{t}\left(H(q)-B(\frac{t}{k})\right).
    \end{equation}

    If we estimate $H(q)$ with $\frac{ct}{k}=\Theta(\frac{1}{k})=\Theta(\epsilon)$ error, we can estimate $H(p)$ with constant error $c$ using Equation~(\ref{eq: shannon-relation}), which requires $\Omega(\sqrt{n}k)=\Omega(\frac{\sqrt{n}}{\epsilon})$ queries to the quantum oracle. $H(q)$ is Shannon entropy of $n+1$-dimensional probability distribution $q$, and estimating it with $\epsilon$-additive error requires $\Omega(\frac{\sqrt{n}}{\epsilon})$. So we can conclude that any algorithm estimating Shannon entropy with $\epsilon$-additive error requires $\Omega(\frac{\sqrt{n}}{\epsilon})$ queries to a quantum probability oracle.
\end{proof}

\section{Discussion} \label{sec: discussion}
The paper establishes a tight bound for estimating the Shannon entropy, up to logarithmic factors. We introduce the singular value separation algorithm to separate the eigenvalues $p_i$ and encode their information into auxiliary control qubits. By applying quantum amplitude amplification and QSVT to the separated quantum state, we efficiently estimate the Shannon entropy within an additive error $\epsilon$, requiring only $\mathcal{\tilde{O}}(\frac{\sqrt{n}}{\epsilon})$ queries to the quantum probability oracle. To prove the lower bound, we construct a quantum oracle where the probability distribution is encoded via Hamming weights in an independent oracle. We conclude that any algorithm outputting $H(p)$ within additive error $\epsilon$ must make at least $\Omega(\frac{\sqrt{n}}{\epsilon})$ queries to the quantum probability oracle.

We anticipate that our algorithmic framework can improve various property testing and estimation problems, such as Rényi and von Neumann entropy estimation. This leads to several open questions for future work:
\begin{itemize}
    \item Can our framework improve the upper bound for von Neumann entropy estimation?
    \item Can it be used to establish tight bounds for Rényi entropy estimation?
    \item Can the advantages of the singular value separation algorithm be leveraged to estimate distance measures such as fidelity, trace distance, and relative entropies?
\end{itemize}

%\begin{figure}[h!]
%	\centering
%	\includegraphics[width=0.48\textwidth]{Gate_based_ML}
%	\caption{\small{\textbf{Scenarios for predicting properties of states from digital quantum computers}}. The left panel shows a typical VQE circuit, where all tunable gates (indicated by orange circles) serve as trainable parameters. The right panel depicts a QNN architecture, where a subset of tunable gates is allocated for encoding classical data, while the remaining gates function as trainable parameters for optimization. The hexagonal and rectangular gates represent Clifford gates.  
%	}
%	\label{fig:vqe-and-qnn}
%\end{figure}

%\begin{figure}
%\refstepcounter{mybox}
%\begin{tcolorbox}[
%title={Box~\themybox: Implicit state reconstruction}, center title, before upper={\normalsize\justifying}]
%\label{mybox:state-reconstruction}
%Implicit state reconstruction refers to the task of learning a generative model that acts as a parametrized distribution $\mathbb{Q}(\bs; \btheta)$, with the goal of optimizing $\btheta$ so that $\mathbb{Q}(\bs; \btheta)$ closely approximates the target distribution $\mathbb{P}(\bs) = \Tr(\rho(\bx) M_s)$ over measurement outcomes $\bs$. Here, $\mathcal{M} = \{M_s\}$ denotes a predefined set of POVM elements, such as the one corresponding to computational basis measurements. This approach enables the model to reproduce the measurement statistics of the quantum state $\rho(\bx)$ without explicitly reconstructing its density matrix.
%\end{tcolorbox}
%\end{figure}

\section*{data availability}
The data supporting the results of this manuscript are given in the article and the appendix. Extra data are available upon reasonable request.

\section*{Acknowledgments}
We acknowledge helpful discussions with Junseo Lee and Mingyu Lee. This work was supported by the National Research Foundation of Korea (NRF) through a grant funded by the Ministry of Science and ICT (Grant No. RS-2025-00515537). This work was also supported by the Institute for Information \& Communications Technology Promotion (IITP) grant funded by the Korean government (MSIP) (Grant Nos. RS-2019-II190003 and RS-2025-02304540), the National Research Council of Science \& Technology (NST) (Grant No. GTL25011-401), and the Korea Institute of Science and Technology Information (KISTI) (Grant No. P25026). 

\section*{Author Contribution}
M.S. conceived and initiated the idea. M.S. and K.J. led and analyzed the main results. All authors wrote and reviewed the manuscript.

\section*{Competing interests}
The authors declare no competing interests.

\newpage

\begin{widetext}

\appendix
\section*{Appendix}

\bigskip
\begin{center}
    {\large\bf Near optimal quantum algorithm for estimating shannon entropy} \\
    \bigskip
    {\large Myeongjin Shin and Kabgyun Jeong}
\end{center}
\subsection{Proof for Theorems and Lemmas}\label{appendix: A}
We give elaborated proofs to Lemma~\ref{lem:W_operator_formal}, Theorem~\ref{thm:sv-separation}, Theorem ~\ref{thm: v}, Theorem~\ref{thm: pick}, and Theorem~\ref{thm: upper}.
\subsubsection{Proof of Lemma~\ref{lem:W_operator_formal}}\label{appendix: W-operator}
Let us recall Lemma~\ref{lem: sv-separation} and set $\widetilde{\Pi},\Pi$ as
\begin{align}
    \widetilde{\Pi}&=\sum_{i=1}^n I\otimes\ket{i}\bra{i}\otimes\ket{i}\bra{i} \\
    \Pi&=\ket{0}\bra{0}\otimes\ket{0}\bra{0}\otimes I,
\end{align}
and $U=O_p$, then $A$ becomes
\begin{equation}
\sum_{i=1}^n\sqrt{p_i}\ket{i}\bra{0}\otimes\ket{\psi_i}\bra{0}\otimes\ket{i}\bra{i}.
\end{equation}
Then, there is a unitary $W(\varphi_j,\epsilon)$ using $\mathcal{O}(\frac{1}{\varphi_j}\log\frac{1}{\epsilon})$ queries to $O_p$ and $O_p^\dagger$ such that
\begin{equation}
    W(\varphi_j,\epsilon)\ket{0}_{C_j}\ket{0}_{P_j}(\ket{0}\ket{0}\ket{i})_{I_j} = \beta_0\ket{0}_{C_j}\ket{\gamma}_{P_j,I_j}+\beta_1\ket{1}_{C_j}\ket{+}_P(\ket{0}\ket{0}\ket{i})_{I_j},
\end{equation}
where $|\beta_0|^2+|\beta_1|^2=1$, such that
\begin{itemize}
    \item if $0\leq\sigma_i\leq\varphi_j$, then $|\beta_1|\leq\epsilon$ and
    \item if $2\varphi_j\leq\sigma_i\leq1$, then $|\beta_0|\leq\epsilon$.
\end{itemize}

Let $\beta_0=\beta_j^{'}(\sqrt{p_i})$ and $\beta_1=\beta_j(\sqrt{p_i})$ then Lemma~\ref{lem:W_operator_formal} is proved.

\subsubsection{Proof of Theorem~\ref{thm:sv-separation}}\label{appendix: sv-separation}
Let us prove that $\ket{\Psi_k}$ defined as
\begin{equation}
    \ket{\Psi_k}=U_k(\epsilon)(\ket{0}_C\ket{0}_P\ket{0}_I\ket{\psi}_{AB})=\prod_{j=1}^k C_jW(\varphi_j,\epsilon)C_jT)(\ket{0}_C\ket{0}_P\ket{0}_I\ket{\psi}_{AB}
\end{equation}
can be represented as
\begin{equation}
    \ket{\Psi_k}=\ket{0}_C\ket{0}_{I_{k+1, m}}\ket{0}_{P_{k+1, m}}\sum_{i=1}^n\sqrt{p_i}B^{'}_k(\sqrt{p_i})\ket{i}_A\ket{\psi_i}_B\ket{\text{garbage}}+\sum_{j=1}^k\ket{\lambda_j}_C\left(\sum_{i=1}^n\sqrt{p_i}B_j(\sqrt{p_i})\ket{i}_A\ket{\psi_i}_B\ket{g_{i,j}}_{P,I}\right).
\end{equation}

Let us use mathematical induction for the proof. For $k=1$, the following holds
\begin{align}
    \ket{\Psi_1} &= C_1W(\varphi_1,\epsilon)C_1T(\ket{0}_C\ket{0}_P\ket{0}_I\ket{\psi}_{AB}) \\
    &= W(\varphi_1,\epsilon)T_1(\ket{0}_C\ket{0}_P\ket{0}_I\ket{\psi}_{AB}) \\
    &= W(\varphi_1,\epsilon)(\ket{0}_C\ket{0}_P\ket{0}_{I_{2,m}}\sum_{i=1}^n\sqrt{p_i}(T_1(\ket{0,0,0}_{I_1}\ket{i}_A))\ket{\psi_i}_{B}) \\
    &= W(\varphi_1,\epsilon)(\ket{0}_C\ket{0}_P\ket{0}_{I_{2,m}}\sum_{i=1}^n\sqrt{p_i}\ket{0,0,i}_{I_1}\ket{i}_A\ket{\psi_i}_{B}) \\
    &= \ket{0}_{C_{2,m}}\ket{0}_{P_{2,m}}\ket{0}_{I_{2,m}}\sum_{i=1}^n\sqrt{p_i}\ket{i}_A\ket{\psi_i}_{B}(W(\varphi_1,\epsilon)\ket{0}_{C_1}\ket{0}_{P_1}\ket{0,0,i}_{I_1})
\end{align}

Then, by applying Lemma~\ref{lem:W_operator_formal}, we have
\begin{align}
    \ket{\Psi_1} &= \ket{0}_{C_{2,m}}\ket{0}_{P_{2,m}}\ket{0}_{I_{2,m}}\sum_{i=1}^n\sqrt{p_i}\ket{i}_A\ket{\psi_i}_{B}(\beta_1^{'}(\sqrt{p_i})\ket{0}_{C_1}\ket{\gamma}_{P_1,I_1}+\beta_1(\sqrt{p_i})\ket{1}_{C_1}\ket{+}_{P_1}(\ket{0,0,i})_{I_j}) \\
    &= \ket{0}_{C}\ket{0}_{P_{2,m}}\ket{0}_{I_{2,m}}\sum_{i=1}^n\sqrt{p_i}\beta_1^{'}(\sqrt{p_i})\ket{i}_A\ket{\psi_i}_{B}\ket{\gamma}_{P_1,I_1} + \ket{1}_{C_1}\sum_{i=1}^n\sqrt{p_i}\beta_1(\sqrt{p_i})\ket{i}_A\ket{\psi_i}_{B}\ket{\text{garbage}}_{P,I},
\end{align}
which proves the $k=1$ case.

Next, suppose that the $k-1$ case holds. Then,
\begin{equation}\label{eq:induction}
    \ket{\Psi_{k-1}}=\ket{0}_C\ket{0}_{I_{k, m}}\ket{0}_{P_{k, m}}\sum_{i=1}^n\sqrt{p_i}B^{'}_{k-1}(\sqrt{p_i})\ket{i}_A\ket{\psi_i}_B\ket{\text{garbage}}+\sum_{j=1}^{k-1}\ket{\lambda_j}_C(\sum_{i=1}^n\sqrt{p_i}B_j(\sqrt{p_i})\ket{i}_A\ket{\psi_i}_B\ket{g_{i,j}}_{P,I}).
\end{equation}

Let us prove the $k$ case. We can easily show the following.
\begin{align}\label{eq:induc-relation}
    \ket{\Psi_{k}} &= C_kW(\varphi_k,\epsilon)C_kT\ket{\Psi_{k-1}},
\end{align}

above $C_kW(\varphi_k,\epsilon)C_kT\ket{\Psi_{k-1}}$ only acts when all qubits in the register $C_1,C_2,\dots,C_{k-1}$ are $\ket{0}$. So, $C_kW(\varphi_k,\epsilon)C_kT\ket{\Psi_{k-1}}$ only acts to
\begin{align}
    \ket{0}_C\ket{0}_{I_{k, m}}\ket{0}_{P_{k, m}}\sum_{i=1}^n\sqrt{p_i}B^{'}_{k-1}(\sqrt{p_i})\ket{i}_A\ket{\psi_i}_B\ket{\text{garbage}} \\
    = \ket{0}_{C_{1,k-1}}\ket{0}_{C_{k+1,m}}\ket{0}_{I_{k+1, m}}\ket{0}_{P_{k+1, m}}\sum_{i=1}^n\sqrt{p_i}B^{'}_{k-1}(\sqrt{p_i})\ket{0}_{C_k}\ket{0}_{P_k}\ket{0,0,0}_{I_k}\ket{i}_A\ket{\psi_i}_B\ket{\text{garbage}}
\end{align}
in Equation~(\ref{eq:induction}). 

Since $W(\varphi_k,\epsilon)$ acts on $(C_k,P_k,I_k)$ and $T_k$ acts on $(I_k,A)$, we focus on the state
\begin{align}
    & W(\varphi_k,\epsilon)T_k\sum_{i=1}^n\sqrt{p_i}B^{'}_{k-1}(\sqrt{p_i})\ket{0}_{C_k}\ket{0}_{P_k}\ket{0,0,0}_{I_k}\ket{i}_A\ket{\psi_i}_B \\
    &= W(\varphi_k,\epsilon)\sum_{i=1}^n\sqrt{p_i}B^{'}_{k-1}(\sqrt{p_i})\ket{0}_{C_k}\ket{0}_{P_k}\ket{0,0,i}_{I_k}\ket{i}_A\ket{\psi_i}_B \\
    &= \sum_{i=1}^n\sqrt{p_i}B^{'}_{k-1}(\sqrt{p_i})(\beta_k^{'}(\sqrt{p_i})\ket{0}_{C_k}\ket{\gamma}_{P_k,I_k}+\beta_k(\sqrt{p_i})\ket{1}_{C_k}\ket{+}_{P_k}(\ket{0,0,i})_{I_k})\ket{i}_A\ket{\psi_i}_B \\
    &= \sum_{i=1}^n\sqrt{p_i}B^{'}_{k-1}(\sqrt{p_i})\beta_k^{'}(\sqrt{p_i})\ket{0}_{C_k}\ket{\gamma}_{P_k,I_k}\ket{i}_A\ket{\psi_i}_B + \sum_{i=1}^n\sqrt{p_i}B^{'}_{k-1}(\sqrt{p_i})\beta_k(\sqrt{p_i})\ket{1}_{C_k}\ket{+}_{P_k}(\ket{0,0,i})_{I_k}\ket{i}_A\ket{\psi_i}_B \\
    &= \ket{0}_{C_k}\sum_{i=1}^n\sqrt{p_i}B^{'}_{k}(\sqrt{p_i})\ket{i}_A\ket{\psi_i}_B\ket{\text{garbage}_1}_{P_k,I_k} + \ket{1}_{C_k}\sum_{i=1}^n\sqrt{p_i}B_{k}(\sqrt{p_i})\ket{i}_A\ket{\psi_i}_B\ket{g_{i,k}}_{P_k,I_k}. 
\end{align}

Finally, integrating the above equation into Equations~(\ref{eq:induction}),~(\ref{eq:induc-relation}), we have
\begin{align}
    \ket{\Psi_{k}}&=\ket{0}_C\ket{0}_{I_{k+1, m}}\ket{0}_{P_{k+1, m}}\sum_{i=1}^n\sqrt{p_i}B^{'}_{k}(\sqrt{p_i})\ket{i}_A\ket{\psi_i}_B\ket{\text{garbage}} + \ket{\lambda_k}_{C}\sum_{i=1}^n\sqrt{p_i}B_{k}(\sqrt{p_i})\ket{i}_A\ket{\psi_i}_B\ket{g_{i,k}}_{P,I} \\ 
    & +\sum_{j=1}^{k-1}\ket{\lambda_j}_C(\sum_{i=1}^n\sqrt{p_i}B_j(\sqrt{p_i})\ket{i}_A\ket{\psi_i}_B\ket{g_{i,j}}_{P,I}) \\
    &= \ket{0}_C\ket{0}_{I_{k+1, m}}\ket{0}_{P_{k+1, m}}\sum_{i=1}^n\sqrt{p_i}B^{'}_{k}(\sqrt{p_i})\ket{i}_A\ket{\psi_i}_B\ket{\text{garbage}}+\sum_{j=1}^{k}\ket{\lambda_j}_C(\sum_{i=1}^n\sqrt{p_i}B_j(\sqrt{p_i})\ket{i}_A\ket{\psi_i}_B\ket{g_{i,j}}_{P,I}).
\end{align}
So, the case $k$ holds. By mathematical induction, we conclude the proof.

\subsubsection{Proof of Lemma~\ref{lem: mag}}\label{appendix: mag}
Since $x\in[\varphi_j,\varphi_{j-1})$ by Lemma~\ref{lem:W_operator_formal} we have
\begin{itemize}
    \item $\beta_1(x)^2,\dots,\beta_{j-1}(x)^2\leq\epsilon^2$ and
    \item $\beta_{m+1}(x)^2\geq 1-\epsilon^2$.
\end{itemize}
So, $B_{j+1}(x)^2=(\prod_{i=1}^{j}\beta_i^{'}(x)^2)\beta_{j+1}(x)^2=\prod_{i=1}^{j}\beta_i^{'}(x)^2-O(\epsilon^2)$. Then, 
\begin{equation}
    B_j(x)^2+B_{j+1}(x)^2 = (\prod_{i=1}^{j-1}\beta_i^{'}(x)^2)(\beta_j(x)^2+\beta_j^{'}(x)^2) + \mathcal{O}(\epsilon^2) = \prod_{i=1}^{j-1}\beta_i^{'}(x)^2 + \mathcal{O}(\epsilon^2).
\end{equation}
Since $\beta_1(x)^2,\dots,\beta_{j-1}(x)^2\leq\epsilon^2$, we have $\beta_1^{'}(x)^2,\dots,\beta_{j-1}^{'}(x)^2\geq1-\epsilon^2$. Then,
\begin{equation}
    B_j(x)^2+B_{j+1}(x)^2 = \prod_{i=1}^{j-1}(1-\epsilon^2) + \mathcal{O}(\epsilon^2) = 1-\mathcal{O}(j\epsilon^2).
\end{equation}

\subsubsection{Proof of Lemma~\ref{lem: Uke}}\label{appendix: Uke}
Each query to $W(\varphi_j,\epsilon)$ requires 

\begin{equation}
    \mathcal{O}(\frac{1}{\varphi_j}\log\frac{1}{\epsilon})=\mathcal{O}(2^j\log\frac{1}{\epsilon})
\end{equation}
queries to $O_p$ and $O_p^\dagger$ (Lemma~\ref{lem: sv-separation}), which is also equivalent to querying $C_jW(\varphi_j,\epsilon)$. So, summing of $j=1,2,\dots,k$, we get 

\begin{equation}
    \sum_{j=1}^k\mathcal{O}(2^j\log\frac{1}{\epsilon})=\mathcal{O}(2^k\log\frac{1}{\epsilon}).
\end{equation}

\subsubsection{Proof of Theorem~\ref{thm: v}}\label{appendix: v}
\begin{proof}
    Suppose that $x\in[\varphi_j,\varphi_{j-1})$, then
    \begin{itemize}
        \item $B_j(x)^2+B_{j+1}^2(x)=1-\mathcal{O}(m\delta^2)$,
        \item $B_1(x)^2+\dots+B_{j-1}(x)^2+B_{j+2}(x)^2+\dots+B_m(x)^2=\mathcal{O}(m\delta^2)$,
        \item $\left|4S_j(\frac{x}{2})^2\log\frac{1}{\varphi_{j+1}}-\log\frac{2}{x}\right|\leq \eta$, $\left|4S_{j+1}(\frac{x}{2})^2\log\frac{1}{\varphi_{j+2}}-\log\frac{2}{x}\right|\leq \eta$. (Definition~\ref{def: poly-S})
    \end{itemize}

By using the above relations, we deduce
\begin{align}
    \left|4\sum_{k=1}^m S_k(x)^2B_k(x)^2\log\frac{1}{\varphi_{k+1}}-\log\frac{2}{x}\right| 
    &= \left|\sum_{k\in[m]/\{j,j+1\}} 4S_k(x)^2B_k(x)^2\log\frac{1}{\varphi_{k+1}}+\sum_{k=j}^{j+1} 4S_k(x)^2B_k(x)^2\log\frac{1}{\varphi_{k+1}}-\log\frac{2}{x}\right| \\
    &\leq \mathcal{O}(m\delta^2)\sum_{k\in[m]/\{j,j+1\}} 4S_k(x)^2\log\frac{1}{\varphi_{k+1}} + \left|\sum_{k=j}^{j+1} B_k(x)^2(\log\frac{2}{x}+\mathcal{O}(\eta))-\log\frac{2}{x}\right| \\
    &\leq \mathcal{O}(m\delta^2)\log\frac{1}{\varphi_{m+1}}+\mathcal{O}(\eta) = \mathcal{\tilde{O}}(m^2\delta^2+\eta).
\end{align}

Suppose that $x\in[0,\varphi_{m})$, then $x\leq\frac{1}{2^m}$. So finally we deduce
\begin{align}
    v &= 2\sum_{i=1}^np_i\left(\sum_{k=1}^m 4S_k(\sqrt{p_i})^2B_k(\sqrt{p_i})^2\log\frac{1}{\varphi_{k+1}}\right)-1 \nonumber\\
    &= 2\sum_{\sqrt{p_i}\geq\varphi_m} p_i\left(\log\frac{2}{\sqrt{p_i}}-1+\mathcal{\tilde{O}}(m\delta^2+\eta)\right) + \sum_{\sqrt{p_i}<\varphi_m}p_i\mathcal{\tilde{O}}\left(\log\frac{1}{\varphi_{m+1}}\right) \nonumber \\
    &= H(p)+\mathcal{\tilde{O}}\left(m\delta^2+\eta+\frac{n}{2^m}\right).
\end{align}
\end{proof}

\subsubsection{Proof of Theorem~\ref{thm: pick}}\label{appendix: pick}
Suppose that $\ket{\Psi_k}$ is prepared and we measure the qubit register $C_k$ with the computational basis. The measurement outputs $\ket{1}_{C_k}$ with 
\begin{equation}
    \text{Sum}(k) = \sum_{i=1}^n p_iB_k(\sqrt{p_i})^2
\end{equation}
probability. The post-measurement state of $\ket{\Psi_k}$ becomes
\begin{equation}
    \ket{\phi_k}=\ket{1}_{C_k}\sum_{i=1}^n\frac{\sqrt{p_i}B_k(\sqrt{p_i})}{\sqrt{\text{Sum}(k)}}\ket{i}_A\ket{\psi_i}_B\ket{\text{garbage}}.
\end{equation}

Let $\ket{S}=\ket{\Psi_k}$ and $\ket{T}=\ket{\phi_k}$, then
\begin{equation}
    \ket{S}=\sqrt{\text{Sum}(k)}\ket{T}+\sqrt{1-\text{Sum}(k)}\ket{\tilde{T}}
\end{equation}
for some $\ket{\tilde{T}}$. The qubit register $C_k$ of $\ket{\tilde{T}}$ is $\ket{0}_{C_k}$. There exists a unitary $U$ such that $U\ket{1}_{C_k}\ket{b}=\ket{1}_{C_k}\ket{b\oplus1}$ and $U\ket{0}_{C_k}\ket{b}=\ket{1}_{C_k}\ket{b}$ for $\inn{T}{\tilde{T}}=0$. Then, we have
\begin{align}
    U\ket{T}\ket{b}&=\ket{T}\ket{b\oplus1}, \;\text{and} \\ 
    U\ket{\tilde{T}}\ket{b}&=\ket{\tilde{T}}\ket{b}.
\end{align}

So we can apply the quantum amplitude amplification (Lemma ~\ref{lem: qaa}) to $\ket{\Psi_k}$ and construct the unitary $V$ satisfying 
\begin{align}
& V \ket{0} = \ket{\phi'}\ket{0}, \\
& \big|\inn{\phi'}{\phi_k}\big|^2 \geq 1 - \delta^2,
\end{align}
using 
\begin{equation}
    \mathcal{O}\left(\log\left(\frac{1}{\delta}\right)\frac{1}{\sqrt{\text{Sum}(k)}}\right)
\end{equation}
queries to $U_k(\epsilon)$.

By Lemma~\ref{lem: Uke}, we can query to $U_k(\epsilon)$ by using
\begin{equation}
    \mathcal{O}\left(2^k\log\frac{1}{\epsilon}\right)
\end{equation}
queries to $O_p$ and $O_p^\dagger$. 

So we conclude that, the unitary $V$ can be implemented thorough
\begin{equation}
    \mathcal{\tilde{O}}\left(\log\left(\frac{1}{\delta}\right)\frac{2^k}{\sqrt{\text{Sum}(k)}}\right)
\end{equation}
queries to $O_p$ and $O_p^\dagger$.

\subsubsection{Proof of Theorem~\ref{thm: upper}}\label{appendix: upper}
Note that we can estimate the Shannon entropy $H(p)$ with $v$, which are represented as
\begin{equation}
    v=-2+\sum_{k=1}^m 8v_k,
\end{equation}
where $v_k$ is
\begin{equation}
    v_k=\sum_{i=1}^n p_iS_k(\sqrt{p_i})^2B_k(\sqrt{p_i})^2\log\frac{1}{\varphi_{k+1}}.
\end{equation}

So our task is to estimate $v_1,v_2,\dots,v_m$. To estimate the value $v_k$, we employ the unitary $U_{S_k}^{(SV)},C_{\widetilde{\Pi}\otimes\ket{+}\bra{+}}NOT$ (in Equation~\ref{eq: QSVT-U-SV}) to $\ket{\phi_k}$ (in Equation~\ref{eq: pick}):

\begin{align}\label{eq: appendix-QSVT-amp-SV}
    &(C_{\widetilde{\Pi}\otimes\ket{+}\bra{+}}NOT)(U_{S_k}^{(SV)}\otimes I)(\ket{0}\ket{0}\ket{\phi_k}\ket{+}\ket{0}_F) \nonumber \\ &=\ket{1}_{C_k}\sum_{i=1}^n\frac{\sqrt{p_i}S_k(\sqrt{p_i})B_k(\sqrt{p_i})}{\sqrt{\text{Sum}(k)}}\ket{i}_A\ket{\psi_i}_B\ket{\text{garbage}_1}\ket{1}_F \nonumber \\ 
    &+\ket{1}_{C_k}\ket{\text{garbage}_2}\ket{0}_F.
\end{align}

The probability of measuring Equation~\ref{eq: appendix-QSVT-amp-SV} with $\ket{1}_{C_k}\ket{1}_F$ is $\sum_{i=1}^n p_iS_k(p_i)^2B_k(p_i)^2$, multiplying $\log\frac{1}{\varphi_{k+1}}$ we obtain $v_k$.

The degree of $S_k$ is $\mathcal{\tilde{O}}(\frac{1}{\varphi_{k+1}})$ by Definition~\ref{def: poly-S}. Hence, we can construct one query to $U_{S_k}^{(SV)}$ with $\mathcal{\tilde{O}}(\frac{1}{\varphi_{k+1}})=\mathcal{\tilde{O}}(2^k)$ queries to $O_p$ and $O_p^\dagger$. Also, we can construct $\ket{\phi_k}$ with $\mathcal{\tilde{O}}(\frac{2^k}{\sqrt{\text{Sum}(k)}})$ queries to $O_p, O_p^\dagger$ by Theorem~\ref{thm: pick}. So, adding the required number of queries, we can conclude that equation~\ref{eq: appendix-QSVT-amp-SV} can be obtained with 
\begin{equation}
    \mathcal{\tilde{O}}\left(\frac{2^k}{\sqrt{\text{Sum}(k)}}\right)
\end{equation}
queries to $O_p$ and $O_p^\dagger$. 

We can apply quantum amplitude estimation to Equation~(\ref{eq: appendix-QSVT-amp-SV}) and obtain the value
\begin{equation}\label{eq: mul1}
    \frac{1}{\text{Sum}(k)}\sum_{i=1}^n p_iS_k(\sqrt{p_i})^2B_k(\sqrt{p_i})^2
\end{equation}
within additive error $\frac{\epsilon}{2m\text{Sum}(k)}$ from  as setting the ``answer" state to $\ket{1}_{C_k}\ket{1}_F$ using
\begin{equation}
    \mathcal{\tilde{O}}\left(\frac{2^k}{\sqrt{\text{Sum}(k)}}\frac{m\text{Sum(k)}}{\epsilon}\right)=\mathcal{\tilde{O}}\left(\frac{m2^k\sqrt{\text{Sum}(k)}}{\epsilon}\right)
\end{equation}
queries to $O_p$ and $O_p^\dagger$. This is because to apply quantum amplitude estimation within additive error $\frac{\epsilon}{2m\text{Sum}(k)}$, it requires $\mathcal{O}\left(\frac{m\text{Sum}(k)}{\epsilon}\right)$ queries to the unitaries that constructs Equation~(\ref{eq: appendix-QSVT-amp-SV}).

Also quantum amplitude estimation can be employed to obtain the value
\begin{equation}\label{eq: mul2}
    \text{Sum}(k)
\end{equation} 
within additive error $\frac{\epsilon}{2m}$ from $\ket{\Psi_k}$ in Equation~(\ref{eq: sv}) as setting the ``answer" state to $\ket{1}_{C_k}$ using 
\begin{equation}
    \mathcal{\tilde{O}}\left(\frac{m2^k\sqrt{\text{Sum}(k)}}{\epsilon}\right)
\end{equation}
queries to $O_p$ and $O_p^\dagger$. This is because $U_k(\epsilon)$ query complexity is $\mathcal{\tilde{O}}(2^k)$ and quantum amplitude estimation requires $\mathcal{O}\left(\frac{m\sqrt{\text{Sum}(k)}}{\epsilon}\right)$ queries to $U_k(\epsilon)$ and its inverse. 

Multiplying the above estimated values Equation~\ref{eq: mul1},\ref{eq: mul2} and $\log\frac{1}{\varphi_{k+1}}$ deduces
\begin{equation}
    v_k=\sum_{i=1}^n p_iS_k(\sqrt{p_i})^2B_k(\sqrt{p_i})^2\log\frac{1}{\varphi_{k+1}}.
\end{equation}
within additive error $\frac{\epsilon}{m}$.

Summing of $v_k$ for $k=1,2,\dots,m$, we finally get
\begin{equation}
    v=-2+\sum_{k=1}^m 8v_k.
\end{equation}
within additive error $\epsilon$. Since $v$ is an adequate approximation of $H(p)$, thus the complexity of estimating $v$ and $H(p)$ within additive error $\epsilon$ is equivalent as proved in Theorem~\ref{thm: v}. 

Therefore, $\text{Sum}(k)$ is bounded by
\begin{align}
    \text{Sum}(k)=\sum_{i=1}^np_iB_k(\sqrt{p_i})^2 \leq \sum_{i=1}^n \left(\frac{1}{2^{k-2}}\right)^2\leq \frac{4n}{4^k},
\end{align}
because $B_k(x)<\mathcal{O}(m\delta^2)$ when $x>\frac{1}{2^{k-2}}$.

The total complexity is 
\begin{equation}
    \mathcal{\tilde{O}}\left(\sum_{k=1}^m\frac{m2^k\sqrt{\text{Sum}(k)}}{\epsilon}\right)=\mathcal{\tilde{O}}\left(\sum_{k=1}^m\frac{m\sqrt{n}}{\epsilon}\right)=\mathcal{\tilde{O}}\left(\frac{\sqrt{n}m^2}{\epsilon}\right)=\mathcal{\tilde{O}}\left(\frac{\sqrt{n}}{\epsilon}\right).
\end{equation}
Because $m=\mathcal{\tilde{O}}(\log\frac{\epsilon}{n})$, and we neglect the logarithmic factors.

\end{widetext}

\bigskip


\begin{thebibliography}{16}

%
\bibitem{bromiley2004shannon}
P. Bromiley, N. Thacker, and E. Bouhova-Thacker, 
``Shannon entropy, R\'{e}nyi entropy, and information," 
{\em Stat. Inf. Ser.} (2004--004) vol. 9, no. 2004, pp. 2--8 (2004).
%
\bibitem{gilyen2019distributional} 
A. Gily\'{e}n and T. Li,
``Distributional property testing in a quantum world,"
\href{https://arxiv.org/abs/1902.00814}{
arXiv:1902.00814 (2019)}.
%
\bibitem{van2021quantum}
J. van Apeldoorn,
``Quantum Probability Oracles \& Multidimensional Amplitude Estimation,"
\newblock In {\em 16th Conference on the
Theory of Quantum Computation, Communication and Cryptography (TQC 2021)}, 
\href{http://dx.doi.org/10.4230/LIPIcs.TQC.2021.9}{
Leibniz Int'l Proc. Inf. (LIPIcs) pp. 9:1--9:11}, 
Schloss Dagstuhl, Leibniz-Zentrum f\"{u}r Informatik, 2021.
%
\bibitem{dagum2000optimal}
P. Dagum, R. Karp, M. Luby, and S. Ross,
``An Optimal Algorithm for Monte Carlo Estimation,"
\href{http://dx.doi.org/10.1137/S0097539797315306}{
SIAM J. Comput.~\textbf{29}, 1484 (2000)}.
%
\bibitem{brassard2000quantum}
G. Brassard, P. H{\o}yer, M. Mosca, and A. Tapp,
``Quantum amplitude amplification and estimation,"
\newblock In {\em Quantum Computation and Quantum Information, Samuel J. Lomonaco, Jr. (editor)}, AMS
\href{https://doi.org/10.1090/conm/305}{
Contemp. Math.~\textbf{305}, 53 (2003)}.
%
\bibitem{tang2025amplitude} 
E. Tang and J. Wright,
``Amplitude amplification and estimation require inverses,"
\href{https://arxiv.org/abs/2507.23787}{
arXiv:2507.23787 (2025)}.
%
\bibitem{wu2016minimax}
Y. Wu and P. Yang,
``Minimax Rates of Entropy Estimation on Large Alphabets via Best Polynomial Approximation,"
\href{http://dx.doi.org/10.1109/TIT.2016.2548468}{
IEEE Tran. inf. Theory~\textbf{62}, 3702 (2016)}.
%
\bibitem{li2018quantum}
T. Li and X. Wu,
``Quantum Query Complexity of Entropy Estimation,"
\href{http://dx.doi.org/10.1109/TIT.2018.2883306}{
IEEE Tran. inf. Theory~\textbf{65}, 2899 (2019)}.
%
\bibitem{montanaro2015quantum}
A. Montanaro,
``Quantum speedup of Monte Carlo methods,"
\href{http://dx.doi.org/10.1098/rspa.2015.0301}{
Proc. R. Soc. A~\textbf{471}, 20150301 (2015)}.
%
\bibitem{wang2024quantum}
X. Wang, S. Zhang, and T. Li,
``A Quantum Algorithm Framework for Discrete Probability Distributions With Applications to R\'{e}nyi Entropy Estimation,"
\href{http://dx.doi.org/10.1109/TIT.2024.3382037}{
IEEE Tran. inf. Theory~\textbf{70}, 3399 (2024)}.
%
\bibitem{bun2018polynomial}
M. Bun, R. Kothari, and J. Thaler,
``The polynomial method strikes back: Tight quantum query bounds via dual polynomials,"
\newblock In {\em Proceedings of the 50th Annual ACM SIGACT Symposium on Theory of Computing}, pp. 297--310, 2018.
%
\bibitem{yoder2014fixed}
T. J. Yoder, G. H. Low, and I. L. Chuang,
``Fixed-Point Quantum Search with an Optimal Number of Queries,"
\href{http://dx.doi.org/10.1103/PhysRevLett.113.210501}{
\prl~\textbf{113}, 210501 (2014)}.
%
\bibitem{gilyen2019quantum}
A. Gily\'{e}n, Y. Su, G. H. Low, and N. Wiebe,
``Quantum singular value transformation and beyond: exponential improvements for quantum matrix arithmetics,"
\newblock In {\em Proceedings of the 51st Annual ACM SIGACT Symposium on Theory of Computing}, pp. 193--204, 2019.
%
\bibitem{wang2024new}
Q. Wang, J. Guan, J. Liu, Z. Zhang, and M. Ying,
``New Quantum Algorithms for Computing Quantum Entropies and Distances,"
\href{http://dx.doi.org/10.1109/TIT.2024.3399014}{
IEEE Tran. inf. Theory~\textbf{70}, 5653 (2024)}.
%
\bibitem{beals2001quantum}
R. Beals, H. Buhrman, R. Cleve, M. Mosca, and R. De Wolf,
``An Optimal Algorithm for Monte Carlo Estimation,"
\href{http://dx.doi.org/10.1145/502090.502097}{
J. ACM~\textbf{48}, 778 (2001)}.
%
\bibitem{kimmel2012quantum}
S. Kimmel,
``Quantum Adversary (Upper) Bound,"
\newblock In {\em International Colloquium on Automata, Languages, and Programming}, pp. 557--568, Springer, 2012.



\end{thebibliography}
\end{document}